\documentclass{gretsi}
\usepackage{amsmath,amssymb,amsfonts}
\usepackage{enumerate}
\usepackage{algorithm,algorithmic}
\usepackage{textcomp}
\usepackage{xcolor}

\usepackage{graphicx}      
\usepackage{natbib}        

\newtheorem{prop}{Proposition}

\newtheorem{proof}{Proof}

\let\classAND\AND
\let\AND\relax
\usepackage{algorithmic}

\let\AND\classAND

\begin{document}
\title{Energy-efficient transmission policies\\ for the linear quadratic control\\ of scalar systems}
\author{\coord{Yifei}{Sun}{1},
	\coord{Samson}{Lasaulce}{2},
	\coord{Michel}{Kieffer}{1},
	\coord{Romain}{Postoyan}{2},
	\coord{Dragan}{Neši\'{c}}{3}
}

\address{\affil{1}{Université Paris-Saclay - CNRS - CentraleSupélec - L2S,\\ F-91192 Gif-sur-Yvette, France}
	\affil{2}{CRAN, CNRS-Université de Lorraine, F-54000 Nancy, France}
	\affil{3}{Department of Electrical and Electronic Engineering, The University of Melbourne, Parkville, 3010, Victoria, Australia.}}
\email{}
\englishabstract{This paper considers controlled scalar systems relying on a lossy wireless feedback channel. In contrast with the existing literature, the focus
	is not on the system controller but on the wireless transmit power
	controller that is implemented at the system side for reporting the
	state to the controller. Such a problem may be of interest, \emph{e.g.}, for the remote control of drones, where communication costs may have to be considered. Determining
	the power control policy that minimizes the combination of the dynamical system
	cost and the wireless transmission energy is shown to be a non-trivial
	optimization problem. It turns out that the recursive structure of
	the problem can be exploited to determine the optimal power control
	policy. As illustrated in the numerical performance analysis, in the scenario of a dynamics without perturbations, the
	optimal power control policy consists in decreasing the transmit power at the right pace. This allows a significant performance gain
	compared to conventional policies such as the full transmit power policy or the open-loop policy.}
\maketitle

\section{Introduction}

\label{sec:introduction}

The dominant paradigm in system control theory is to assume that information
exchanges between the controller(s) and the system(s) to be controlled
are perfect. When information exchanges occur over wireless channels,
this assumption may be questionable and even not realistic at all.
This is one of the reasons why there is an active research area at
the interface between control theory and wireless communications.
Among representative research works of this approach, we can quote
the following papers. The problem of imperfect communication
between the various components of a system is addressed in \citet{Hespanha07}.
In \citet{Delchamps90}, the problem of imperfect feedback is considered
in the case where the noise is caused by the quantization of transmitted
data. In \citet{Shi13}, it is shown how a finite communication data
rate impacts the controller design. The impact of fast fading wireless
channel fluctuations on the control design has been addressed, \emph{e.g.},
in \citet{Gatsis14}, \citet{Varma2020TAC}, and \citet{Balaghiinaloo2020ECC}. The coexistence of several controlled systems sharing the same communication channel prone to interference is considered
in \citet{Gatsis16}.

In the present paper, in contrast with the existing literature, the main
technical focus is not on the system controller but on the control
of the wireless transmit power implemented at the system side for
reporting its state to the controller through a wireless feedback
channel. This scenario may be of interest, \emph{e.g.}, in the remote
control of drones, when the control input is evaluated by a remote
controller from measurements of the state of the drone that are transmitted
over a wireless channel. The approach proposed has at least three
salient features. First, the transmit power is adapted to the wireless
feedback channel statistics and the system state. Second, the objective
pursued consists of a combination of a system control objective and
a communication objective (namely, the wireless transmission energy); managing the wireless transmit power is both relevant in terms of consumed energy and electromagnetic pollution. Third, this adaptation is performed in the presence of an additive
perturbation on the (linear) system dynamics and a multiplicative
noise for the wireless feedback channel (which corresponds to data
packet erasures). This complete framework has not been addressed yet
in the literature even in the simple case of scalar linear systems.
Good representatives of the closest literature are \citet{willems1976feedback,besson2000approximate,Primbs09}
where the authors also assume a multiplicative noise model for the
communication channel but do not focus on the wireless transmit power
control problem by both pursuing a system control objective and a
wireless transmission energy objective. Rather the cited papers focus
on the problem of system stability.

The paper is structured as follows. In Section~\ref{sec:problem},
the technical problem to be solved is formulated. Determining the
best power control policy is shown to amount to solving a non-trivial
multilinear problem. To solve it, we resort to an iterative search
technique described in Section~\ref{sec:with}. Then, in Section~\ref{sec:Numerical-experiments},
we conduct a numerical performance analysis to illustrate the benefits
of controlling properly the wireless transmit power. Conclusions and
perspectives are provided in Section~\ref{sec:conclusion}.

\section{Assumptions and Problem Formulation}

\label{sec:problem}

We consider a dynamical system whose state is scalar. The motivation
behind this assumption is that the problem appears to be of interest
even in this simple case and this makes the introduction of the proposed
framework clearer. The system state is denoted by $x\in\mathbb{R}$
and assumed to evolve according to the discrete-time state
equation 
\begin{equation}
x_{t+1}=ax_{t}+bu_{t}+d_{t},\label{eq:perfect}
\end{equation}
where $t \in \{1,...,T\}$, $T \geqslant 1$ being the considered time horizon, $\left(a,b\right)\in\mathbb{R}^{2}$, $u_{t}\in\mathbb{R}$
is the control input, and $d_{t}\sim\mathcal{N}\left(0,\mathbf{\sigma}_{\text{d}}^{2}\right)$
is a Gaussian random state perturbation. One assumes that the random variables $x_{t}$
and $d_{t'}$ are independent for all $t'\geqslant t$.

\begin{figure}[h]
\centering \includegraphics[scale=0.85]{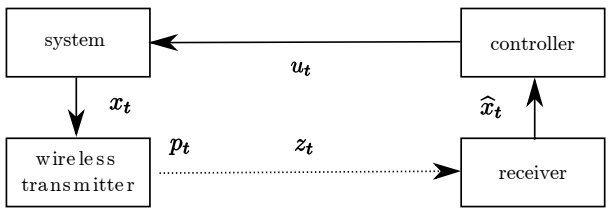} \caption{Communication and control setup}

\label{fig:system} 
\end{figure}

At time (or for data packet) $t$, the system transmits its state $x_{t}$ to a remote controller in charge of computing the control input $u_{t}$, see
Figure~\ref{fig:system}. This transmission is performed with power
$p_{t}$ over a wireless communication channel modeled by a classical baseband additive communication channel of the form $y_t = h_t s_t + w_t$, where $s_t \in \mathbb{C}$ represents the coded signal, $h_t \in \mathbb{C}$ is the channel coefficient, and $w_t \sim \mathcal{N}(0,\sigma^2)$ the i.i.d. Gaussian communication noise. From this communication model, we will only exploit the two following key quantities: $g_t = |h_t|^2$ and the communication noise variance $\sigma^2$. In this paper, no particular assumptions are made on the channel i.i.d. random process $(g_t)_{t \in \{1,...,T\}}$ except in Prop. 1 and 5 and for simulations, where we assume an exponential p.d.f. with mean $\overline{g}$ (that is, a classical Rayleigh fading model for $h_t$). The message is assumed to be successfully
received when the signal-to-noise (SNR) ratio at the receiver is sufficient
to allow error-free decoding. This occurs with probability 
\begin{equation}
\begin{aligned}\pi\left(p_{t}\right)= & \Pr\left[\frac{g_{t}p_{t}}{\sigma^{2}}\geqslant\gamma\right],\end{aligned}
\label{eq:snr}
\end{equation}
where $g_{t}$ is the channel gain at time $t$, $\sigma^{2}$ the variance of the (additive and zero-mean Gaussian) communication noise, and $\gamma$ is the
SNR threshold. Therefore, the controller receives 
\begin{equation}
\widehat{x}_{t}=x_{t}z_{t}\label{eq:noise}
\end{equation}
where $z_{t}\sim\text{Ber}\left(\pi_{t}\right)$ is a realization
of a Bernoulli random variable with parameter 
\begin{equation}
\pi\left(p_{t}\right)=\Pr\left[z_{t}=1\right].\label{eq:SuccessRate}
\end{equation}
When the message is too noisy to be decoded successfully, we have
that $z_{t}=0$, which corresponds in practice to a data packet loss.
In what follows, the (packet) success probability $\pi\left(p_{t}\right)$
is denoted by $\pi_{t}$ to make the notations simpler.

We assume that $z_{t}$ is known at the receiver (this is typical
when the communication system uses a cyclic redundancy check to verify
the integrity of the received message). When $z_{t}=1$, a static
feedback is evaluated as 
\begin{equation}
u_{t}  =k\widehat{x}_{t}=kx_{t}z_{t}
\label{eq:control}    
\end{equation}
where $k\in\mathbb{R}$. When $z_{t}=0$, the controller is unable
to acquire the state $x_{t}$ and the chosen control input is $u_{t}=0$, as proposed,
\emph{e.g.}, in \citet{Schenato09,Quevedo14}. The control input is
then fed back to the system. We assume that the channel from the controller
to the system is perfect. This motivation behind this assumption is twofold. It it fully relevant in communications scenarios where there is an asymmetry in terms of resources (e.g., in terms of transmit power, bandwidth, or computation resources). It also makes our analysis more tractable and easier to interpret for a first step into the direction taken in this paper.

We assume that $x_{1}\sim\mathcal{N}\left(0,\mathbf{\sigma}_{\text{x}}^{2}\right)$
and consider an optimal controller with finite horizon $T$. The considered
problem is to find a transmission power policy $p_{1:T}=\left(p_{1},p_{2},\dots,p_{T}\right)^{\text{T}}$,
to be applied over the control horizon, that minimizes 
\begin{equation}
\overline{J}_{1:T}\left(p_{1:T}\right)=\mathbb{E}_{z_{1:T},d_{1:T}}\left[\sum_{t=1}^{T}\left(qx_{t}^{2}+ru_{t}^{2}+p_{t}\right)\right]\label{eq:goal}
\end{equation}
with bounded transmission power 
\begin{equation}
0\leqslant p_{t}\leqslant P_{\text{max}},t=1,\dots,T,\label{eq:PowerConstraint}
\end{equation}
where $\mathbb{E}_{z_{1:T},d_{1:T}}[\cdot]$ indicates that the expectation is performed with respect to $z_1,\dots,z_T$ and $d_1,\dots,d_T$, $q >0$, and $r>0$. The expectation is performed with respect to $z_{1:T}$, which depends
on the transmission powers, 
and with respect to
$d_{1:T}$. One has $z_{i}$ and $d_{j}$ independent for all $i=1,\dots,T$
and $j=1,\dots,T$. The presence of the wireless transmission energy
term in the above cost allows the transmission to be energy-efficient (for more details see, \emph{e.g.}, \cite{lasaulce2011game}). Technically, the presence of this term makes the problem non-trivial. Indeed, without
any energy cost associated with the wireless system, the problem boils
down to a classical finite-horizon LQR (linear quadratic regulation) problem \citet{shaiju2008formulas}
and the cost is trivially minimized by transmitting at full power
all the time. Because of the presence of the wireless transmission
energy cost, a tradeoff needs to be found between the conventional
system cost and the communication cost associated with the feedback
channel. Technically, the formulated optimization problem turns out
to be non-trivial as explained in the next section, which proposes
an iterative numerical technique to determine the
optimal transmission power control policy.

\section{Proposed solution}
\label{sec:with}

First, let us reformulate the optimization problem associated with
\eqref{eq:goal} and \eqref{eq:PowerConstraint}. This is the purpose
of the following proposition.
\begin{prop}
\label{prop:COP} Assume that the i.i.d. random process $(g_t)_{t \in \{1,...,T\}}$ follows a Rayleigh fading law with mean $\mathbb{E}(g_t) = \overline{g}$. Denote by $\pi_{1:T}=\left(\pi_{1},...,\pi_{T}\right)^{\text{T}}$
the sequence of success probabilities. The problem of minimizing $\overline{J}_{1:T}\left(p_{1:T}\right)$
with respect to $p_{1:T}$ under the constraints \eqref{eq:PowerConstraint}
can be reformulated as 
\begin{align}
\min_{\pi_{1:T}} & \ C\left(\pi_{1:T}\right)\label{eq:OptimProblem}\\
\text{s.t.} & -\frac{\gamma\sigma^{2}}{ \overline{g} \ln\pi_{t}}-P_{\text{max}}\leqslant0,\text{ }t=1\dots T\nonumber 
\end{align}
with 
\begin{align}
C\left(\pi_{1:T}\right) & =\mathbf{\sigma}_{\text{x}}^{2}\left(q+rk^{2}\pi_{1}\right)\nonumber \\
 & +\mathbf{\sigma}_{\text{x}}^{2}\sum_{t=2}^{T}\left(q+rk^{2}\pi_{t}\right)\prod_{i=1}^{t-1}\left(a^{2}+\left(b^{2}k^{2}+2abk\right)\pi_{i}\right)\nonumber \\
 &\hspace{-1cm} +\sigma_{\text{d}}^{2}\sum_{t=2}^{T}\left(q+rk^{2}\pi_{t}\right)\sum_{i=1}^{t-1}\prod_{r=i+1}^{t-1}\left(a^{2}+\left(b^{2}k^{2}+2abk\right)\pi_{i}\right)\nonumber \\
 & -\sum_{\text{\ensuremath{t=1}}}^{T}\frac{\gamma\sigma^{2}}{\overline{g} \ln\pi_{t}}.\label{eq:OptimSingle_noise}
\end{align}
\end{prop}

\begin{proof}
See Appendix \ref{subsec:Proof-of-COP}.
\end{proof}
From Proposition \ref{prop:COP}, we see that the cost function $C\left(\pi_{1:T}\right)$
is multilinear w.r.t. the success probability vector $\pi_{1:T}$,
which means that, in general, the cost is neither linear, convex,
nor quasi-convex. The corresponding problem is therefore non-trivial, see
\citet{bao2015global,Yan07}. As will be shown in what follows, the
recursive structure of the problem can be exploited to determine the
optimal sequence of probabilities of success and therefore the optimal
sequence of transmit power levels. For that purpose, one decomposes
\eqref{eq:goal} as 
\begin{equation}
\overline{J}_{1:T}\left(p_{1:T}\right)=\overline{J}_{1:t-1}\left(p_{1:t-1}\right)+\overline{J}'_{t:T}\left(p_{1:T}\right)\label{eq:twoparts}
\end{equation}
where 
\[
\overline{J}_{1:t-1}\left(p_{1:t-1}\right)=\mathbb{E}_{z_{1:t-1},d_{1:t-1}}\left(\sum_{\ell=1}^{t-1}\left(qx_{\ell}^{2}+rk^{2}x_{\ell}^{2}z_{\ell}+p_{\text{\ensuremath{\ell}}}\right)\right)
\]
and 
\[
\overline{J}'_{t:T}\left(p_{1:T}\right)=\mathbb{E}_{z_{1:T},d_{1:T}}\left(\sum_{\ell=t}^{T}\left(qx_{\ell}^{2}+rk^{2}x_{\ell}^{2}z_{\ell}+p_{\ell}\right)\right).
\]

Furthermore, Proposition \ref{prop:Cost_tT} separates the terms which contains $p_t$ (or $\pi_t$) from $\overline{J}'_{t:T}\left(p_{1:T}\right)$.
\begin{prop}
\label{prop:Cost_tT}In \eqref{eq:twoparts}, $\overline{J}'_{t:T}\left(p_{1:T}\right)$
is expressed as 
\begin{align}
\overline{J}'_{t:T}\left(p_{1:T}\right) =&\mathbb{E}_{z_{1:t-1}d_{1:t-1}}\left[x_{t}^{2}\right]\overline{F}\left(p_{t:T}\right)\nonumber \\
 & +\sigma_{\text{d}}^{2}\overline{F}_{\text{s}}\left(p_{t+1:T}\right)+\sum_{\text{\ensuremath{\ell=t}}}^{T}p_{\ell},\label{eq:Jprime}
\end{align}
where, for all $t<T$ 
\begin{align*}
\overline{F}\left(p_{t:T}\right) & =\left(q+rk^{2}\pi_{t}\right)\\
 & +\sum_{\ell=t+1}^{T}\left(q+rk^{2}\pi_{\ell}\right)\prod_{i=t}^{\ell-1}\left(a^{2}+\left(2abk+b^{2}k^{2}\right)\pi_{i}\right)
\end{align*}
and for all $t<T-1$
\begin{align*}
\overline{F}_{\text{s}}\left(p_{t+1:T}\right) & =\sum_{\ell=t+1}^{T}\left(q+rk^{2}\pi_{\ell}\right)\\
 & \times\sum_{i=t}^{\ell-1}\prod_{r=i+1}^{\ell-1}\left(a^{2}+\left(2abk+b^{2}k^{2}\right)\pi_{r}\right).
\end{align*}
\end{prop}
\begin{proof}
See Appendix \ref{subsec:Proof-of-Cost_tT}. 
\end{proof}
In \eqref{eq:Jprime}, $\overline{F}\left(p_{t:T}\right)$ depends on $p_t$ (or $\pi_t$) while $\overline{F}_{\text{s}}\left(p_{t+1:T}\right)$ is independent of them. These two terms can be evaluated by Proposition \ref{prop:recursion}.
\begin{prop}
\label{prop:recursion}$\overline{F}\left(p_{t:T}\right)$ and $\overline{F}_{s}\left(p_{t+1:T}\right)$
can be evaluated using the following backward recursions
\[
\overline{F}\left(p_{t:T}\right)=\left(q+rk^{2}\pi_{t}\right)+\left(a^{2}+\pi_{t}\left(2abk+b^{2}k^{2}\right)\right)\overline{F}\left(p_{t+1:T}\right)
\]
for all $t\leqslant T-1$ and 
\[
\overline{F}_{\text{s}}\left(p_{t:T}\right)=\overline{F}\left(p_{t:T}\right)+\overline{F}_{\text{s}}\left(p_{t+1:T}\right)
\]
for all $t\leqslant T-2$.
\end{prop}

\begin{proof}
See Appendix \ref{subsec:Proof-of-recursion}.
\end{proof}
These backward recursions are initialized considering the transmission
power $p_{T}$ minimizing \eqref{eq:goal}.
\begin{prop}
\label{prop:bwrecursion}The transmission power $p_{T}$ at time $T$
minimizing \eqref{eq:goal} is $p_{T}=0$ and leads to
\[
\overline{F}\left(p_{T}\right)=q
\]
and
\[
\overline{F}_{s}\left(p_{T}\right)=q.
\]
\end{prop}

\begin{proof}
See Appendix \ref{subsec:Proof-of-bwrecursion}.
\end{proof}

We can then determine $\pi_t$ minimizing \eqref{eq:OptimSingle_noise} when $\pi_{t'}$ is fixed for all $t'=1\dots,T,\ t\neq t'$. This is shown in Proposition~\ref{prop:minimum}.

\begin{prop}
\label{prop:minimum} Assume a Rayleigh fading law with mean $\overline{g}$ for $g_t$. Consider some $t\in\left\{ 1,\dots,T-1\right\} $
and assume that $\pi_{t'}$ is fixed for all $t'=1\dots,T,$ $t'\neq t$.
The value of $\pi_{t}$ minimizing \eqref{eq:OptimSingle_noise} with
the constraint \eqref{eq:PowerConstraint-1} is either $\pi_{t}=0$
or $\pi_{t}=\min\left(e^{-\frac{\gamma\sigma^{2}}{P_{\max}\overline{g}}},\pi^{0}\right)$,
where $\pi^{0}$ is such that $e^{-2}<\pi^{0}$ and 
\[
\mathbb{E}_{z_{1:t-1}d_{1:t-1}}\left[x_{t}^{2}\right]\frac{\partial}{\partial\pi_{t}}\overline{F}\left(p_{t:T}\right)+\text{\ensuremath{\frac{\gamma\sigma^{2}}{\pi^{0}\ln^{2}\pi^{0}\overline{g}}}}=0.
\]
\end{prop}

\begin{proof}
See Appendix \ref{subsec:Proof-of-minimum}.

Consider a transmission power policy $p_{1:T}^{\left(0\right)}$ and
its corresponding $\pi_{1:T}^{\left(0\right)}$. From Proposition~\ref{prop:minimum},
for all $t=1,\dots,T$, one can obtain
\[
\pi_{t}^{\star}=\arg\min_{\pi_{t}\in\mathcal{I}}C\left(\pi_{1}^{\left(0\right)},\dots,\pi_{t-1}^{\left(0\right)},\pi_{t},\pi_{t+1}^{\left(0\right)},\dots,\pi_{T}^{\left(0\right)}\right)
\]
where $\mathcal{I}=\left\{0,\min\left(e^{-\frac{\gamma\sigma^{2}}{P_{\text{max}}\overline{g}}},\pi^{0}\right)\right\}$.

A set 
\[
\mathcal{P}=\left\{ \left[\pi_{1}^{\star},\pi_{2}^{\left(0\right)},\dots,\pi_{T}^{\left(0\right)}\right],\dots,\left[\pi_{1}^{\left(0\right)},\dots,\pi_{T-1}^{\left(0\right)},\pi_{T}^{\star}\right]\right\} 
\]
of associated success vectors is obtained. The vector
\begin{align*}
\pi_{1:T}^{\left(1\right)}= & \arg\min_{\pi_{1:T}\in\mathcal{P}}C\left(\pi_{1:T}\right)
\end{align*}
and the associated transmission power policy $p_{1:T}^{\left(1\right)}$
provides a reduced cost. The above process may be repeated as illustrated
in Algorithm~\ref{Alg:algOneAgent} to obtain an improved transmission
power policy.
\end{proof}
\begin{algorithm}
\caption{Transmission power optimization}
\label{Alg:algOneAgent}

\begin{algorithmic}
\global\long\def\algorithmicrequire{\textbf{Input:}}%
\REQUIRE Time horizon $T$, $\sigma_{\text{x}}^{2}=\mathbb{E}\left[x_{1}^{2}\right]$,
$a,\ b,\ k,\ q$, $r$, $\sigma_{\text{d}}^{2}$;

\STATE Initialization: $p_{1}^{\left(0\right)}=\dots p_{T}^{\left(0\right)}=0$,
$k=1;$

\global\long\def\algorithmicrequire{\textbf{Output:}}%

\REQUIRE Power policy $p_{1:T}$;

\WHILE{$k\leqslant$$k_{\text{max}}$}

\FOR {$t=T:-1:1$}

\STATE Using $p_{1}^{\left(k-1\right)},\dots,p_{T}^{\left(k-1\right)}$,
\eqref{eq:pi}, and $\overline{F}\left(p_{T}\right)=q$, determine
$\overline{F}\left(p_{t:T}\right)$ by backward recursion using Proposition
\ref{prop:recursion};

\ENDFOR

\FOR {$t=1:T$}

\STATE From $p_{t-1}^{\left(k-1\right)}$ and $\mathbb{E}\left[x_{t-1}^{2}\right]$,
evaluate $\mathbb{E}\left[x_{t}^{2}\right]$ using \eqref{eq:xt};

\STATE Determine the minimum of $\frac{\partial\overline{J}}{\partial\pi_{t}}$
using \eqref{eq:dJdpi} obtained at $\pi_{t}=\min\left(e^{-2},e^{-\frac{\gamma\sigma^{2}}{P_{\text{max}}\overline{g}}}\right)$;

\IF {$\left.\frac{\partial\overline{J}}{\partial\pi_{t}}\right|_{\pi_{t}}<0$}

\STATE Determine $\pi^{0}$ and its corresponding power $p^{0}$
using \eqref{eq:pi};

\STATE Determine $\pi_{t}^{*}\in\left\{ 0,\min\left(\pi^{0},e^{-\frac{\gamma\sigma^{2}}{P_{\text{max}}\overline{g}}}\right)\right\} $
minimizing the cost;

\STATE Determine the power $p_{t}^{*}$ corresponding to $\pi_{t}^{*}$
using \eqref{eq:pi};

\ELSE

\STATE $p_{t}^{*}=p_{t}^{\left(k-1\right)}$;

\ENDIF

\ENDFOR

\STATE $\left[p_{1}^{\left(k\right)},\dots,p_{T}^{\left(k\right)}\right]$ is one of the
element of $\mathcal{P}=\left\{ [p_{1}^{*},p_{2}^{\left(k-1\right)},\dots,p_{T}^{\left(k-1\right)}]\right.$,$\dots$,
$\left.[p_{1}^{\left(k-1\right)},p_{2}^{\left(k-1\right)},\dots,p_{T}^{*}]\right\} $
that minimizes the cost \eqref{eq:twoparts}

\ENDWHILE

\end{algorithmic} \label{algo:algorithm} 
\end{algorithm}

In the loop of Algorithm~\ref{Alg:algOneAgent}, each element
of the vector $[p_{1}^{\left(k\right)},\dots,p_{T}^{\left(k\right)}]$
is replaced with its updated version $[p_{1}^{\left(k+1\right)},\dots,p_{T}^{\left(k+1\right)}]$
which induces a smaller cost. The latter property combined with the fact that the cost is bounded guarantees the convergence of the proposed algorithm. The obtained power policy is then obtained by assuming the knowledge of the average power of $x_{1}$, that is $\mathbb{E}\left[x_{1}^{2}\right]$, and not the value of $x_{1}$ itself.

\section{Numerical performance analysis}

\label{sec:Numerical-experiments}

To study the behavior of Algorithm~\ref{Alg:algOneAgent}, consider a system with $a=1.1$, $b=-1$,
and $k=1$, as well as a realization $x_{1}=1$. For the communication, $P_{\text{max}}=3$ and $\gamma\sigma^{2}/\overline{g}=1$.
Moreover, $q=1,\ r=0.5$, $T=30$.

Consider a first scenario with $\sigma_{\text{d}}^{2}=0$ (perturbation-free
case). Figure~\ref{fig:curves} illustrates $p_{t}$ for the considered nominal values of the parameters mentioned before and for
alternate parameter values where a single change of one component
is performed. This illustrates the impact of each parameter on the
transmission power policy.

\begin{figure}[h]
\includegraphics[scale=0.65]{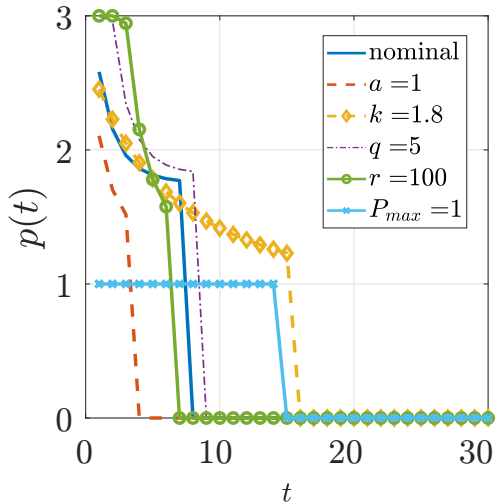}

\caption{Impact of the model parameters on the obtained transmit power control policy in the case of no perturbations on the system dynamics ($\sigma_{\text{d}}^{2}=0$).}
\label{fig:curves} 
\end{figure}

From Figure~\ref{fig:curves}, we observe that for the nominal value
of the parameters, transmissions occur at the beginning and stop at
$t=8$. Decreasing $P_{\max}$ leads to transmissions with less power
at the beginning to stop at $t=15.$ Decreasing $a$ leads to a more
stable open-loop system, requiring less transmissions. Choosing $k=1.8$,
which, in close-loop, is less stable, leads to an increase of communications.

If $q$ increases, the weight of the state in the cost function increases
leading to more control effort. On the contrary, a larger value of
$r$ putting more weight on the transmission costs, reduces the number
of transmissions.

Let $a=1.1$, $b=-1$,  $k=1.8$, $P_{\text{max}}=3$, $\gamma\sigma^{2}/\overline{g}=1$
and $\sigma_{x}^{2}=1$. Moreover, $q=1,\ r=0.5$ and $T=30$. To
illustrate the impact of perturbation, Figure~\ref{fig:curves-1}
shows the power control policy with different values of $\sigma_{d}^{2}$.
\begin{figure}[h]
\includegraphics[scale=0.6]{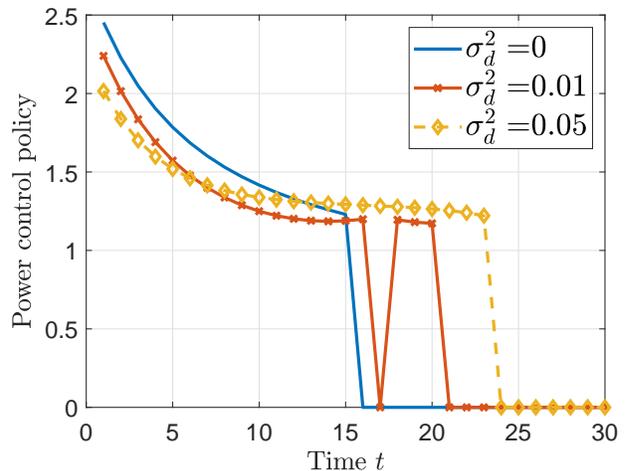}

\caption{Obtained transmission power policy for different values of the variance  $\sigma_{\text{d}}^{2}$ of the state perturbation, when
$T=30$}
\label{fig:curves-1} 
\end{figure}

From Figure \ref{fig:curves-1}, we observe that, when $\sigma_{\text{d}}^{2}$
increases, there are less time slot where $p_{t}\neq0$, indicating
more communication will occur. This is due to the fact that increasing the
perturbation drives the system away from equilibrium and leads an increase
need of communications.

Consider $a=1.1$, $b=-1$, $k=1.8$, $P_{\text{max}}=3$, $\gamma\sigma^{2}/\overline{g}=1$,
$q=1,\ r=0.5$, $\sigma_{x}^{2}=1$ and $\sigma_{d}^{2}=0.05$. Different
values of the time horizon $T$ have been considered. The average
value over $10000$ samples of (\ref{eq:goal}) is compared for three
different policies: sending with full transmit power $P_{\text{max}}=3$,
open loop policy (sending nothing), sending with power determined
by Algorithm~\ref{Alg:algOneAgent}, see Figure~\ref{fig:horizonT}.

\begin{figure}[h]
\includegraphics[scale=0.6]{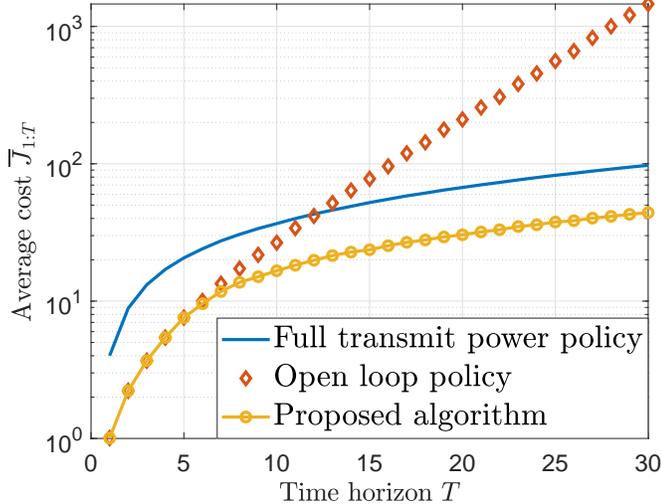}

\caption{Impact of the choice of the power control policy: For $T=30$ using the proposed policy allows the combined cost to be divided by $50$.}
\label{fig:horizonT} 
\end{figure}

Figure \ref{fig:horizonT} shows that, when $T\leqslant6$, the curve
of the proposed algorithm overlaps the one of open loop policy. The reason
is that the communication between the system and the controller is
not worthy, when the control time horizon is not enough large. When $T>6$,
the performance of the power control policies obtained by the proposed algorithm
is the best compared to the algorithm where we always send with full
transmit power and the algorithm with no communication.

\section{Conclusion}

\label{sec:conclusion} In this paper, we consider a scalar system which control input is evaluated by a remote controller from information sent by the system over a noisy wireless channel. We focus on the optimization of the transmit power implemented at the system side for reporting the
state to the controller. We have shown that determining
the power control policy that minimizes the combination of the dynamical system cost and the wireless transmission energy is a non-trivial
optimization problem. We have proposed an iterative algorithm to evaluate a transmission power policy achieving a trade-off between the system control cost and the energy spent for wireless transmission. In absence of perturbation on the system dynamics, the optimal transmit power is seen to be decreasing with time. The power profile depends on the values for the system and control parameters. The obtained profiles differ significantly from the profiles consisting in transmitting at full power or not transmitting at all (open loop scenario). Significant gains 
can be observed when comparing the proposed policy to the aforementioned conventional policies. This work will be extended to the vector case and to situations where the wireless resources have to be shared by several system-controller pairs which may generate interference. A significant extension would be to address the challenging case of non-linear systems, which supposes to revisit the proof techniques used in this paper.

\bibliography{ifacconf.bbl}







\appendix

\section{Proofs}

\subsection{Proof of Proposition \ref{prop:COP}\label{subsec:Proof-of-COP}}

For a Rayleigh fading channel model, one has
\begin{equation}
\pi\left(p_{t}\right) =\int_{\frac{\gamma\sigma^{2}}{p_{t}}}^{+\infty}\frac{1}{\overline{g}}\exp\left(-\frac{g}{\overline{g}}\right)\mathrm{d}g
 =\exp\left(-\frac{\gamma\sigma^{2}}{p_{t}\overline{g}}\right).\label{eq:pi}
\end{equation}
From (\ref{eq:pi}), one observes that searching for $p_{1:T}$
minimizing \eqref{eq:goal} under the power constraints \eqref{eq:PowerConstraint}
is equivalent to searching for $\pi_{1:T}$ minimizing \eqref{eq:goal}
under the constraints 
\begin{equation}
-\frac{\gamma\sigma^{2}}{\ln\pi_{t}\overline{g}}\leqslant P_{\text{max}},t=1,\dots,T.\label{eq:PowerConstraint-1}
\end{equation}
Moreover, combining (\ref{eq:noise}) and (\ref{eq:control}), for
all $t\geqslant1$, one gets 
\begin{align}
x_{t+1} & =\left(a+bkz_{t}\right)x_{t}+d_{t}\nonumber \\
 & =\prod_{\ell=1}^{t}\left(a+bkz_{\ell}\right)x_{1}+\sum_{\ell=1}^{t}\prod_{r=\ell+1}^{t}\left(a+bkz_{r}\right)d_{\ell}\label{eq:xt}
\end{align}
where, by convention $\prod_{r=t+1}^{t}\left(a+bkz_{r}\right)=1.$
From \eqref{eq:xt}, one observes that $x_{t+1}$ depends on $x_{1}$,
on $z_{1},\dots,z_{t}$, and on $d_{1},\ \dots,d_{t}$. Now, since
$z_{t}\sim\text{Ber}\left(\pi_{t}\right)$, one has 
\begin{equation}
\mathbb{E}\left[z_{t}\right] =1\times\Pr\left[z_{t}=1\right]+0\times\Pr\left[z_{t}=0\right]
 =\pi_{t}\label{eq:Ez}
\end{equation}
and similarly, 
\begin{equation}
\mathbb{E}\left[z_{t}^{2}\right]=\pi_{t}.\label{eq:Ez2}
\end{equation}
Moreover, since $z_{t}$ and $z_{t'}$ are independent when $t\neq t'$,
\begin{equation}
\mathbb{E}\left[z_{t}z_{t'}\right]=\pi_{t}\pi_{t'}.\label{eq:Ez3}
\end{equation}
Then, \eqref{eq:OptimSingle_noise} is obtained by introducing (\ref{eq:xt})
in (\ref{eq:goal}), and using (\ref{eq:Ez}), (\ref{eq:Ez2}), and
(\ref{eq:Ez3}) as follows 
\begin{align*}
 & C\left(\pi_{1:T}\right)\\
 & =\mathbb{E}_{z_{1:T}d_{1:T}}\left[\sum_{t=1}^{T}\left(qx_{t}^{2}+rk^{2}x_{t}^{2}z_{t}+p_{t}\right)\right]\\
 & =\mathbb{E}_{z_{1:T}d_{1:T}}\left[\sum_{t=1}^{T}\left(q+rk^{2}z_{t}\right)x_{t}^{2}\right]+\sum_{\text{\ensuremath{t=1}}}^{T}p_{t}\\
 & =\mathbb{E}_{z_{1:T}d_{1:T}}\left[\left(q+rk^{2}z_{1}\right)x_{1}^{2}+\sum_{t=2}^{T}\left(q+rk^{2}z_{t}\right)\times\right.\\
 & \left.\left(\prod_{i=1}^{t-1}\left(a+bkz_{i}\right)x_{1}+\sum_{i=1}^{t-1}d_{i}\prod_{r=i+1}^{t-1}\left(a+bkz_{r}\right)\right)^{2}\right]\\
 & +\sum_{\text{\ensuremath{t=1}}}^{T}p_{t}
\end{align*}
since $d_{i},\ d_{j}$ are independent if $i\neq j$, and $\mathbb{E}\left[d_{\ell}\right]=0$,
\begin{align*}
 & C\left(\pi_{1:T}\right)\\
 & =\mathbb{E}_{z_{1:T}d_{1:T}}\left[\left(q+rk^{2}z_{1}\right)x_{1}^{2}+\sum_{t=2}^{T}\left(q+rk^{2}z_{t}\right)\text{\ensuremath{\times}}\right.\\
 & \left(\text{\ensuremath{x_{1}^{2}\prod_{i=1}^{t-1}\left(a+bkz_{i}\right)^{2}}}+\underbrace{\sum_{i=1}^{t-1}d_{i}^{2}\prod_{r=i+1}^{t-1}\left(a+bkz_{r}\right)^{2}}_{d_{i},\ d_{j}\text{ are independent if }i\neq j}+\right.\\
 & \left.\left.\underbrace{2\left(\prod_{i=1}^{t-1}\left(a+bkz_{i}\right)x_{1}\right)\left(\sum_{i=1}^{t-1}\prod_{r=i+1}^{t-1}\left(a+bkz_{r}\right)d_{i}\right)}_{\mathbb{E}\left[d_{\ell}\right]=0}\right)\right]\\
 & +\sum_{\text{\ensuremath{t=1}}}^{T}p_{t}
\end{align*}
the last term of the expectation is vanishing,
\begin{align*}
 & C\left(\pi_{1:T}\right)\\
 & =\mathbb{E}_{z_{1:T}d_{1:T}}\left[x_{1}^{2}\left(\left(q+rk^{2}z_{1}\right)+\sum_{t=2}^{T}\left(q+rk^{2}z_{t}\right)\times\right.\right.\\
 & \left.\prod_{i=1}^{t-1}\left(a+bkz_{i}\right)^{2}\right)+\sum_{t=2}^{T}\left(q+rk^{2}z_{t}\right)\sum_{i=1}^{t-1}d_{i}^{2}\times\\
 & \left.\prod_{r=i+1}^{t-1}\left(a+bkz_{r}\right)^{2}\right]+\sum_{\text{\ensuremath{t=1}}}^{T}p_{t}\\
 & =\mathbf{\sigma}_{\text{x}}^{2}\left(q+rk^{2}\pi_{1}\right)\\
 & +\mathbf{\sigma}_{\text{x}}^{2}\sum_{t=2}^{T}\left(q+rk^{2}\pi_{t}\right)\prod_{i=1}^{t-1}\left(a^{2}+\left(b^{2}k^{2}+2abk\right)\pi_{i}\right)\\
 & +\sigma_{\text{d}}^{2}\sum_{t=2}^{T}\left(q+rk^{2}\pi_{t}\right)\sum_{i=1}^{t-1}\prod_{r=i+1}^{t-1}\left(a^{2}+\left(b^{2}k^{2}+2abk\right)\pi_{i}\right)\\
 & -\sum_{\text{\ensuremath{t=1}}}^{T}\frac{\gamma\sigma^{2}}{ \overline{g} \ln\pi_{t}}
\end{align*}

\subsection{Proof of Proposition \ref{prop:Cost_tT}\label{subsec:Proof-of-Cost_tT}}

Developing the cost function, one gets
\begin{equation}
\begin{aligned} & \overline{J}_{t:T}\left(p_{1:T}\right)\\
= & \mathbb{E}_{z_{1:T}d_{1:T-1}}\left[\sum_{\ell=t}^{T}\left(qx_{\ell}^{2}+rk^{2}x_{\ell}^{2}z_{\ell}+p_{\ell}\right)\right]\\
= & \mathbb{E}_{z_{1:T}d_{1:T-1}}\left[\left(q+rk^{2}z_{t}\right)x_{t}^{2}+\sum_{\ell=t+1}^{T}\left(q+rk^{2}z_{\ell}\right)\times\right.\\
 & \left.\left(\prod_{i=t}^{\ell-1}\left(a+bkz_{i}\right)x_{t}+\sum_{i=t}^{\ell-1}d_{i}\prod_{r=i+1}^{\ell-1}\left(a+bkz_{r}\right)\right)^{2}\right]+\\
 & \sum_{\text{\ensuremath{\ell=t}}}^{T}p_{\ell}.
\end{aligned}
\end{equation}
Then
\begin{equation}
\begin{aligned} & \overline{J}_{t:T}\left(p_{1:T}\right)\\
= & \mathbb{E}_{z_{1:T}d_{1:T-1}}\left[\left(q+rk^{2}z_{t}\right)x_{t}^{2}+\sum_{\ell=t+1}^{T}\left(q+rk^{2}z_{\ell}\right)\text{\ensuremath{\times}}\right.\\
 & \left(\text{\ensuremath{x_{t}^{2}\prod_{i=t}^{\ell-1}\left(a+bkz_{i}\right)^{2}}}+\sum_{i=t}^{\ell-1}d_{i}^{2}\prod_{r=i+1}^{\ell-1}\left(a+bkz_{r}\right)^{2}+\right.\\
 & \left.\left.2\left(\prod_{i=t}^{\ell-1}\left(a+bkz_{i}\right)x_{t}\right)\left(\sum_{i=t}^{\ell-1}\prod_{r=i+1}^{\ell-1}\left(a+bkz_{r}\right)d_{i}\right)\right)\right]\\
 & +\sum_{\text{\ensuremath{\ell=t}}}^{T}p_{\ell}
\end{aligned}
\label{eq:x_indep_zd}
\end{equation}

Since the expected value of $x_{t}^{2}$ is independent of $z_{t:T}$
and $d_{t:T-1}$, \eqref{eq:x_indep_zd} becomes 
\[
\begin{aligned} & \overline{J}'_{t:T}\left(p_{1:T}\right)\\
= & \mathbb{E}_{z_{1:T}d_{1:T-1}}\left[x_{t}^{2}\left(q+rk^{2}z_{t}\right)+\sum_{\ell=t+1}^{T}\left(q+rk^{2}z_{\ell}\right)\times\right.\\
 & \left.\left(x_{t}^{2}\prod_{i=t}^{\ell-1}\left(a+bkz_{i}\right)^{2}+\sum_{i=t}^{\ell-1}d_{i}^{2}\prod_{r=i+1}^{\ell-1}\left(a+bkz_{r}\right)^{2}\right)\right]\\
 & +\sum_{\text{\ensuremath{\ell=t}}}^{T}p_{\ell}\\
= & \mathbb{E}_{z_{1:T}d_{1:T-1}}\left[x_{t}^{2}\left(\left(q+rk^{2}z_{t}\right)+\sum_{\ell=t+1}^{T}\left(q+rk^{2}z_{\ell}\right)\times\right.\right.\\
 & \left.\prod_{i=t}^{\ell-1}\left(a+bkz_{i}\right)^{2}\right)+\sum_{\ell=t+1}^{T}\left(q+rk^{2}z_{\ell}\right)\times\\
 & \left.\sum_{i=t}^{\ell-1}d_{i}^{2}\prod_{r=i+1}^{\ell-1}\left(a+bkz_{r}\right)^{2}\right]
 +\sum_{\text{\ensuremath{\ell=t}}}^{T}p_{\ell}.
\end{aligned}
\]
Then, introducing 
\[
\begin{aligned} & \overline{F}\left(p_{t:T}\right)\\
= & \mathbb{E}_{z_{t:T}}\left[\left(q+rk^{2}z_{t}\right)+\sum_{\ell=t+1}^{T}\left(q+rk^{2}z_{\ell}\right)\prod_{i=t}^{\ell-1}\left(a+bkz_{i}\right)^{2}\right]
\end{aligned}
\]
and 
\[
\begin{aligned} & \overline{F}_{\text{s}}\left(p_{t+1:T}\right)\\
= & \mathbb{E}_{z_{t+1:T-1}}\left[\sum_{\ell=t+1}^{T}\left(q+rk^{2}z_{\ell}\right)\sum_{i=t}^{\ell-1}\prod_{r=i+1}^{\ell-1}\left(a+bkz_{r}\right)^{2}\right],
\end{aligned}
\]
one obtains \eqref{eq:Jprime}.

\subsection{Proof of Proposition \ref{prop:recursion}\label{subsec:Proof-of-recursion}}

From Proposition~\ref{prop:Cost_tT}, one has that:
\begin{align*}
\overline{J}'_{t:T}\left(p_{1:T}\right) & =\mathbb{E}_{z_{1:t}d_{1:t-1}}\left[\left(q+rk^{2}z_{t}\right)x_{t}^{2}+p_{t}\right]\\
 & +\overline{J}'_{t+1:T}\left(p_{1:T}\right)\\
 & =\mathbb{E}_{z_{1:t}d_{1:t-1}}\left[\left(q+rk^{2}z_{t}\right)x_{t}^{2}\right]\\
 & +\mathbb{E}_{z_{1:t}d_{1:t}}\left[x_{t+1}^{2}\right]\overline{F}\left(p_{t+1:T}\right)\\
 & +\sigma_{d}^{2}\overline{F}_{s}\left(p_{t+2:T}\right)+\sum_{\text{\ensuremath{\ell=t}}}^{T}p_{\ell}.
\end{align*}
Using \eqref{eq:Ez} and \eqref{eq:Ez2}, one obtains: 
\begin{align*}
\overline{J}'_{t:T}\left(p_{1:T}\right) & =\mathbb{E}_{z_{1:t-1}d_{1:t-1}}\left[x_{t}^{2}\right]\left(q+rk^{2}\pi_{t}\right)\\
 & +\mathbb{E}_{z_{1:t}d_{1:t}}\left[\left(\left(a+bkz_{t}\right)x_{t}+d_{t}\right)^{2}\right]\overline{F}\left(p_{t+1:T}\right)\\
 & +\sigma_{\text{d}}^{2}\overline{F}_{s}\left(p_{t+2:T}\right)+\sum_{\text{\ensuremath{\ell=t}}}^{T}p_{\ell}\\
 & =\left(q+rk^{2}\pi_{t}\right)\mathbb{E}_{z_{1:t-1}d_{1:t-1}}\left[x_{t}^{2}\right]\\
 & \hspace{-1.5cm}+\mathbb{E}_{z_{1:t}d_{1:t}}\left[\left(a+bkz_{t}\right)^{2}x_{t}^{2}+d_{t}^{2}+2\left(a+bkz_{t}\right)x_{t}d_{t}\right]\times\\
 & \overline{F}\left(p_{t+1:T}\right)+\sigma_{d}^{2}\overline{F}_{s}\left(p_{t+2:T}\right)+\sum_{\text{\ensuremath{\ell=t}}}^{T}p_{\ell}
\end{align*}
Then 
\[
\begin{aligned}\overline{J}'_{t:T}\left(p_{1:T}\right) & =\left(q+rk^{2}\pi_{t}\right)\mathbb{E}_{z_{1:t-1}d_{1:t-1}}\left[x_{t}^{2}\right]\\
 & \hspace{-1.5cm}+\left(\mathbb{E}_{z_{1:t-1}d_{1:t-1}}\left[x_{t}^{2}\right]\left(a^{2}+\pi_{t}\left(2abk+b^{2}k^{2}\right)\right)+\sigma_{\text{d}}^{2}\right)\times\\
 & \overline{F}\left(p_{t+1:T}\right)+\sigma_{d}^{2}\overline{F}_{s}\left(p_{t+2:T}\right)+\sum_{\text{\ensuremath{\ell=t}}}^{T}p_{\ell}\\
 & =\mathbb{E}_{z_{1:t-1}d_{1:t-1}}\left[x_{t}^{2}\right]\left(\left(q+rk^{2}\pi_{t}\right)+\right.\\
 & \left.\left(a^{2}+\pi_{t}\left(2abk+b^{2}k^{2}\right)\right)\overline{F}\left(p_{t+1:T}\right)\right)\\
 & +\sigma_{\text{d}}^{2}\overline{F}\left(p_{t+1:T}\right)+\sigma_{\text{d}}^{2}\overline{F}_{s}\left(p_{t+2:T}\right)+\sum_{\text{\ensuremath{\ell=t}}}^{T}p_{\ell}\\
 & =\mathbb{E}_{z_{1:t-1}d_{1:t-1}}\left[x_{t}^{2}\right]\overline{F}\left(p_{t:T}\right)\\
 & +\sigma_{\text{d}}^{2}\overline{F}_{s}\left(p_{t+1:T}\right)+\sum_{\text{\ensuremath{\ell=t}}}^{T}p_{\ell}.
\end{aligned}
\]
leading to the backward recursions 
\begin{align*} 
\overline{F}\left(p_{t:T}\right)
= & \left(q+rk^{2}\pi_{t}\right)\\
 &+\left(a^{2}+\pi_{t}\left(2abk+b^{2}k^{2}\right)\right)\overline{F}\left(p_{t+1:T}\right)
\end{align*}
\[
 \overline{F}_{s}\left(p_{t+1:T}\right)
= \overline{F}\left(p_{t+1:T}\right)+\overline{F}_{s}\left(p_{t+2:T}\right).
\]

\subsection{Proof of Proposition \ref{prop:bwrecursion}\label{subsec:Proof-of-bwrecursion}}

Writing (\ref{eq:twoparts}) at $t=T-1$, one gets
\[
\overline{J}_{1:T}\left(p_{1:T}\right)=\overline{J}_{1:T-1}\left(p_{1:T-1}\right)+\overline{J}'_{T:T}\left(p_{1:T}\right).
\]
Since $\overline{J}_{1:T-1}\left(p_{1:T-1}\right)$ does not depend
on $p_{T}$, the value of $p_{T}$ minimizing $\overline{J}_{1:T}\left(p_{1:T}\right)$
has to minimize
\begin{align*}
\overline{J}'_{T:T}\left(p_{1:T}\right) & =\mathbb{E}_{z_{1:T}d_{1:T-1}}\left[qx_{T}^{2}+rk^{2}z_{T}x_{T}^{2}+p_{T}\right]\\
 & =\mathbb{E}_{z_{1:T-1}d_{1:T-1}}\left[x_{T}^{2}\right]\overline{F}\left(p_{T}\right)+p_{T},
\end{align*}
with 
$\overline{F}\left(p_{T}\right)=\mathbb{E}_{z_{T}}\left[q+rk^{2}z_{T}\right]$.
From the above expressions, one sees that $p_{T}=0$ and thus $z_{T}=0$
(absence of transmission) minimizes $\overline{J}'_{T:T}\left(p_{1:T}\right)$.
When $z_{T}=0$, one gets $\overline{F}\left(0\right)=q.$

At time $t=T-1$ 
\begin{align*}
 & \overline{J}'_{T-1:T}\left(p_{1:T}\right)\\
= & \mathbb{E}_{z_{1:T}d_{1:T-1}}\left[qx_{T}^{2}+p_{T}+\right.\\
 & \left.\left(q+rk^{2}z_{T-1}\right)x_{T-1}^{2}+p_{T-1}\right]\\
= & \mathbb{E}_{z_{1:T}d_{1:T-1}}\left[q\left(\left(a+bkz_{T-1}\right)x_{T-1}+d_{T-1}\right)^{2}+\right.\\
 & \left.p_{T}+\left(q+rk^{2}z_{T-1}\right)x_{T-1}^{2}+p_{T-1}\right]\\
= & \mathbb{E}_{z_{1:T}d_{1:T-1}}\left[q\left(a+bkz_{T-1}\right)^{2}x_{T-1}^{2}+qd_{T-1}^{2}\right.\\
 & \left.+p_{T}+\left(q+rk^{2}z_{T-1}\right)x_{T-1}^{2}+p_{T-1}\right]\\
= & \mathbb{E}_{z_{1:T-2}d_{1:T-2}}\left[x_{T-1}^{2}\right]\left(\left(q+rk^{2}\pi_{T-1}\right)+\right.\\
 & \left.\left(a^{2}+\left(b^{2}k^{2}+2abk\right)\pi_{T-1}\right)q\right)+\sigma_{\text{d}}^{2}q+\\
 & \sum_{\ell=T-1}^{T}p_{\ell}.
\end{align*}
Using Proposition~\ref{prop:recursion} with $\overline{F}\left(p_{T}\right)=q$,
one can derive 
\begin{align*}
\overline{F}\left(p_{T-1:T}\right)
= & \left(q+rk^{2}\pi_{T-1}\right)\\
&+\left(a^{2}+\left(2abk+b^{2}k^{2}\right)\pi_{T-1}\right)q
\end{align*}
and
\begin{align*}
 \overline{J}'_{T-1:T}\left(p_{1:T}\right)
= & \mathbb{E}_{z_{1:T-2}d_{1:T-2}}\left[x_{T-1}^{2}\right]\overline{F}\left(p_{T-1:T}\right)+\\
 & \sigma_{\text{d}}^{2}q+\sum_{\ell=T-1}^{T}p_{\ell}.
\end{align*}

From Proposition \ref{prop:Cost_tT}
\begin{align*}
\overline{J}'_{T-1:T}\left(p_{1:T}\right) & =\mathbb{E}_{z_{1:T-2}d_{1:T-2}}\left[x_{T-1}^{2}\right]\overline{F}\left(p_{T-1:T}\right)\\
 & +\sigma_{\text{d}}^{2}\overline{F}_{\text{s}}\left(p_{T}\right)+\sum_{\text{\ensuremath{\ell=T-1}}}^{T}p_{\ell}
\end{align*}
one can obtain
$\overline{F}_{s}\left(p_{T}\right)=q$.

\subsection{Proof of Proposition~\ref{prop:minimum}\label{subsec:Proof-of-minimum}}

To determine the value of $\pi_{t}$ which minimizes (\ref{eq:goal}),
consider (\ref{eq:twoparts}) and evaluate
\[
\frac{\partial\overline{J}_{1:T}}{\partial\pi_{t}}=\frac{\partial}{\partial\pi_{t}}\left(\overline{J}_{1:t-1}\left(p_{1:t-1}\right)+\overline{J}'_{t:T}\left(p_{1:T}\right)\right).
\]
Since $\overline{J}_{1:t-1}\left(p_{1:t-1}\right)$ does not depend
on $\pi_{t}$, using Proposition \ref{prop:Cost_tT}, one obtains 
\begin{align}
\frac{\partial\overline{J}_{1:T}}{\partial\pi_{t}} & =\frac{\partial}{\partial\pi_{t}}\overline{J}'_{t:T}\left(p_{1:T}\right)\nonumber \\
 & =\mathbb{E}_{z_{1:t-1}d_{1:t-1}}\left[x_{t}^{2}\right]\frac{\partial}{\partial\pi_{t}}\overline{F}\left(p_{t:T}\right)\nonumber \\
 & +\frac{\partial p_{t}}{\partial\pi_{t}}+\sigma_{\text{d}}^{2}\frac{\partial}{\partial\pi_{t}}\overline{F}_{s}\left(p_{t+1:T}\right)\nonumber \\
 & =\mathbb{E}_{z_{1:t-1}d_{1:t-1}}\left[x_{t}^{2}\right]\left(rk^{2}+\left(2abk+b^{2}k^{2}\right)\overline{F}\left(p_{t+1:T}\right)\right)\nonumber \\
 & +\frac{\gamma\sigma^{2}}{\pi_{t}  \overline{g} \ln^{2}\pi_{t}}.\label{eq:dJdpi}
\end{align}
We can determine $\overline{F}\left(p_{t+1:T}\right)$ using Proposition~\ref{prop:recursion}.
We also need evaluate $\mathbb{E}_{z_{1:t-1}d_{1:t-1}}\left[x_{t}^{2}\right]$
by forward recursion. From \eqref{eq:xt}, we have 
\begin{align}
 & \mathbb{E}_{z_{1:t}d_{1:t}}\left[x_{t+1}^{2}\right]\nonumber\\
= & \mathbb{E}_{z_{1:t}d_{1:t}}\left[(a+bkz_{t})^{2}x_{t}^{2}+d_{t}^{2}+2d_{t}(a+bkz_{t})x_{t}\right]\nonumber \\
= & \left(a^{2}+\pi_{t}\left(b^{2}k^{2}+2abk\right)\right)\mathbb{E}_{z_{1:t-1}d_{1:t-1}}\left[x_{t}^{2}\right]+\sigma_{\text{d}}^{2}.\label{eq:x2t}
\end{align}
The noise does not appear explicitly in \eqref{eq:dJdpi}. But from
\eqref{eq:x2t}, the additive noise still affect on the derivation
$\frac{\partial\overline{J}_{1:T}}{\partial\pi_{t}}$ through $\mathbb{E}_{z_{1:t-1}}\left[x_{t}^{2}\right]$.

We have a function 
\[
\begin{aligned}\frac{\partial\overline{J}_{1:T}}{\partial\pi_{t}}= & \mathbb{E}_{z_{1:t-1}d_{1:t-1}}\left[x_{t}^{2}\right]\times\\
 & \left(rk^{2}+\left(2abk+b^{2}k^{2}\right)\overline{F}\left(p_{t+1:T}\right)\right)+
 \frac{\gamma\sigma^{2}}{\pi_{t} \overline{g} \ln^{2}\pi_{t}}.
\end{aligned}
\]
From (\ref{eq:PowerConstraint-1}), one has 
\[
0\leqslant\pi_{t}\leqslant e^{-\frac{\gamma\sigma^{2}}{P_{\text{max}}\overline{g}}}\leqslant1.
\]

The first component of $\frac{\partial\overline{J}_{1:T}}{\partial\pi_{t}}$
is independent of $\pi_{t}$ and the second component $\text{\ensuremath{\frac{\gamma\sigma^{2}}{\pi_{t}  \overline{g} \ln^{2}\pi_{t}}}}$
is a function of $\pi_{t}$. Consequently, $\frac{\partial\overline{J}}{\partial\pi_{t}}$
is minimum when $\ensuremath{\frac{\gamma\sigma^{2}}{\pi_{t} \overline{g} \ln^{2}\pi_{t}}}$
is minimum. The derivative
\begin{equation}
\frac{\partial}{\partial\pi_{t}}\left(\ensuremath{\frac{\gamma\sigma^{2}}{\pi_{t} \overline{g} \ln^{2}\pi_{t}}}\right)=\frac{-\ln^{2}\pi_{t}\overline{g}-2\ln\pi_{t}\overline{g}}{\left(\pi_{t} \overline{g} \ln^{2}\pi_{t} \right)^{2}}\label{eq:derivative}
\end{equation}
vanishes when$\ln\pi_{t}=0$ or $\ln\pi_{t}=-2$, \emph{i.e.}, when
$\pi_{t}=1$ or $\pi_{t}=e^{-2}$. When $\pi_{t}\in\left]0,e^{-2}\right[$,
\eqref{eq:derivative} is negative and when $\pi_{t}\in\left]e^{-2},1\right[$,
\eqref{eq:derivative} is positive. Thus the minimum of $\frac{\partial\overline{J}_{1:T}}{\partial\pi_{t}}$
is obtained for $\pi_{t}=e^{-2}$ if $e^{-2}<\exp(-\frac{\gamma\sigma^{2}}{P_{\text{max}}\overline{g}})$.
Else, the minimum of $\frac{\partial\overline{J}_{1:T}}{\partial\pi_{t}}$
is obtained for $\pi_{t}=e^{-\frac{\gamma\sigma^{2}}{P_{\text{max}}\overline{g}}}$.
The minimum is obtained for $\pi'=\min\left(e^{-2},e^{-\frac{\gamma\sigma^{2}}{P_{\text{max}}\overline{g}}}\right)$.

Assume that the minimum value of the derivative $\partial\overline{J}_{1:T}\slash\partial\pi_{t}|_{\pi_{t}=\pi'}$,
is negative, $\exists\delta$, $\mathcal{A}=\left]\pi'-\delta,\pi'\right]$
such that, $\pi'\in\mathcal{A}$, $\partial\overline{J}_{1:T}\slash\partial\pi_{t}<0$
leading to a decrease of $\overline{J}_{1:T}$.

Assume that the minimum value of $\partial\overline{J}_{1:T}\slash\partial\pi_{t}$
is negative and is obtained when $\pi_{t}=e^{-\frac{\gamma\sigma^{2}}{P_{\text{max}}\overline{g}}}$.
The minimum of $\overline{J}_{1:T}$ over the interval $\left]0,e^{-\frac{\gamma\sigma^{2}}{P_{\text{max}}\overline{g}}}\right]$
is then either obtained when $\pi_{t}=0$ or when $\pi_{t}=e^{-\frac{\gamma\sigma^{2}}{P_{\text{max}}\overline{g}}}$.

Assume now that the minimum value of $\partial\overline{J}_{1:T}\slash\partial\pi_{t}$
is negative and obtained when $\pi_{t}=e^{-2}<e^{-\frac{\gamma\sigma^{2}}{P_{\text{max}}\overline{g}}}$.
The minimum of $\overline{J}_{1:T}$ over the interval $\left]0,e^{-\frac{\gamma\sigma^{2}}{P_{\text{max}}\overline{g}}}\right]$
is then obtained when $\pi_{t}=0$, $\pi_{t}=e^{-\frac{\gamma\sigma^{2}}{P_{\text{max}}\overline{g}}}$,
or $\pi_{t}=\pi^{0}$, where $\pi^{0}$ is such that $e^{-2}<\pi^{0}<e^{-\frac{\gamma\sigma^{2}}{P_{\text{max}}\overline{g}}}$
and $\mathbb{E}_{z_{1:t-1}}\left[x_{t}^{2}\right]\frac{\partial}{\partial\pi_{t}}\left(\overline{F}\left(p_{t},\dots,p_{T}\right)\right)+\text{\ensuremath{\frac{\gamma\sigma^{2}}{\pi^{0} \overline{g} \ln^{2}\pi^{0}}}}=0$. 
\end{document}